\documentclass[letterpaper, 10 pt, conference]{ieeeconf}  % Comment this line out if you need a4paper

\IEEEoverridecommandlockouts                              % This command is only needed if
\usepackage{amsmath,amssymb,amsfonts}
\usepackage{mathrsfs}
\usepackage{epsfig}
\usepackage{cite}
\usepackage[lined,boxed,commentsnumbered,ruled]{algorithm2e}
\usepackage{algorithmic}
\usepackage{subfigure}
\usepackage{enumerate}
\usepackage{comment}
\usepackage{color}
\usepackage{url}
\usepackage{booktabs}
\usepackage{setspace}
\usepackage{threeparttable}
\usepackage{balance}
\usepackage{verbatim}
\usepackage{epstopdf}

\newtheorem{definition}{Definition}
\newtheorem{lemma}{Lemma}
\newtheorem{theorem}{Theorem}

\newtheorem{assumption}{Assumption}
\newtheorem{remark}{Remark}
\newtheorem{corollary}{Corollary}

\title{\LARGE \bf
Learning Optimal Robust Control of Connected Vehicles \\ in Mixed Traffic Flow
}

\author{Jie Li, 
        Jiawei Wang,
        Shengbo Eben Li,
        Keqiang Li% <-this % stops a space
\thanks{This study is supported by National Key R\&D Program of China with 2022YFB2502901. It is also partially supported by the National Natural Science Foundation of China under grant number 52221005. All correspondence should be sent to S. Li.}% <-this % stops a space
\thanks{J.~Li, J.~Wang, S.~Li and K.~Li are with the School of Vehicle and Mobility, Tsinghua University, Beijing, China. (\{jie-li18,wang-jw18\}@mails.tsinghua.edu.cn, \{lishbo,likq\}@tsinghua.edu.cn).}%
}

\begin{document}

\maketitle
\thispagestyle{empty}
\pagestyle{empty}

%%%%%%%%%%%%%%%%%%%%%%%%%%%%%%%%%%%%%%%%%%%%%%%%%%%%%%%%%%%%%%%%%%%%%%%%%%%%%%%%
\begin{abstract}

%Mixed traffic systems have gradually come to reality and attracted the attention of academic circles with the wide application of connected and automated vehicle (CAV) technology. However, the instability caused by external disturbance through mixed traffic flow is a challenging problem. 
% This paper employs the policy iteration method to learn the $H_{\infty}$ optimal controller with the smallest $L_2$-gain for nonlinear zero-sum games of mixed traffic systems. Given an $L_2$-gain, solving the Hamilton–Jacobi inequality with Hamiltonian constraint in policy iteration paradigm is proved to derive stabilizing controller with desired attenuation performance. Based on the updated robust controller, the $L_2$-gain is optimized in sum of squares program by applying the gap of Hamiltonian constraint. Simulation studies verify that the obtained controller enables the CAVs to dampen traffic perturbations and smooth traffic flow. 
Connected and automated vehicles (CAVs) technologies promise to attenuate undesired traffic disturbances. However, in mixed traffic where human-driven vehicles (HDVs) also exist, the nonlinear human-driving behavior has brought critical challenges for effective CAV control. 
This paper employs the policy iteration method to learn the optimal robust controller for nonlinear mixed traffic systems. Precisely, we consider the $H_{\infty}$ control framework and formulate it as a zero-sum game, the equivalent condition for whose solution is converted into a Hamilton–Jacobi inequality with a Hamiltonian constraint. Then, a policy iteration algorithm is designed to generate stabilizing controllers with desired attenuation performance. Based on the updated robust controller, the attenuation level is further optimized in sum of squares program by leveraging the gap of the Hamiltonian constraint. Simulation studies verify that the obtained controller enables the CAVs to dampen traffic perturbations and smooth traffic flow.

%This conversion allows for simplifying the solving process and further optimizing the attenuation performance. 
% guarantees given an attenuation level. 
 %It is proved that our method is able to derive stabilizing controller with desired attenuation performance. 
  %, verifying the effectiveness of the proposed method.

\end{abstract}

%%%%%%%%%%%%%%%%%%%%%%%%%%%%%%%%%%%%%%%%%%%%%%%%%%%%%%%%%%%%%%%%%%%%%%%%%%%%%%%%
\section{Introduction}

Undesired traffic disturbances may easily lead to the occurrence of traffic waves, where the involved vehicles periodically accelerate and decelerate, resulting in decreased travel efficiency, fuel economy and driving safety~\cite{sugiyama2008traffic}. The emergence of connected and automated vehicles (CAVs) %enabling multi-vehicle or vehicle-infrastructure cooperation, 
promises efficient attenuation of traffic disturbances~\cite{li2020robust}. Recent research has either theoretically or empirically revealed that in mixed traffic, where human-driven vehicles (HDVs) also exist, CAVs could mitigate traffic waves and stabilize traffic flow even in a low penetration rate~\cite{stern2018dissipation,wang2020controllability,zheng2020smoothing}. %communications with other CAVs or human-driven vehicles (HDVs), highway infrastructure, and mobile devices, have a great potential to attenuate traffic disturbances~\cite{li2020robust}. In practice, with the increasing penetration rate of autonomous driving vehicles, mixed traffic systems containing both CAVs and HDVs have gradually become the study object of controller design. 

Regarding the specific control methods of CAVs, existing model-based research mostly relies on a linearized dynamics model for the mixed traffic system. Common modeling frameworks include Lagrangian control~\cite{molnar2020open,cui2017stabilizing}, connected cruise control (CCC)~\cite{jin2017optimal}, and leading cruise control (LCC)~\cite{wang2022leading}. To obtain such a model, these research usually needs to linearize a car-following model of HDVs, \emph{e.g.}, optimal velocity model (OVM)~\cite{bando1995dynamical}, around certain traffic equilibrium state. In practice, however, the performance of these methods may easily be compromised due to the nonlinear human-driving behaviors and the time-varying traffic equilibrium states. To address these issues, some model-free learning policies have been recently proposed via reinforcement learning (RL)~\cite{wu2021flow,shi2021connected} or data-driven predictive control~\cite{wang2022deep,wang2023distributed}, but the dependence on large-scale traffic data has limited its practical deployment. 

To our best knowledge, very few studies have addressed the disturbance attenuation problem in nonlinear mixed traffic systems, with a very recent exception in~\cite{wang2023general}, where Lyapunov methods are employed for stability analysis. Besides, existing methods have not considered optimizing the disturbance attenuation performance of CAVs in mixed traffic flow. To optimize and control the nonlinear traffic system via CAVs, approximate dynamic programming (ADP) provides a promising technique through solving the nonlinear $H_{\infty}$ optimal control problem of the mixed traffic system. Particularly, this work utilizes the tool of policy iteration (PI) from ADP to learn a robust optimal controller for mixed traffic systems with explicit consideration of nonlinear human-driving behaviors. 

PI is a class of effective numerical method for nonlinear robust control~\cite{abu2006policy, zhu2017policy}, and has been recently applied to a wide range of diverse fields; see, \emph{e.g.}, %coal gasification process~\cite{wei2013adaptive}
robot manipulator~\cite{kong2020robust} and vehicle platooning~\cite{zhu2018adaptive}. Compared with RL, PI enjoys complete theoretical foundations, including algorithm convergence and closed-loop stability~\cite{li2022reinforcement}, which play a critical role in connected vehicle control. % The latter theoretical guarantee is particularly important for traffic flow control, in which the controller drives each vehicle near the expected headway and velocity. 
Indeed, benefiting from the advantages of this method, an adaptive optimal controller has been recently designed in~\cite{huang2020learning} with respect to unknown and heterogeneous HDV behaviors in mixed traffic. However, how to achieve an optimal disturbance attenuation performance remains an open question. %Unexpected traffic disturbance exists widely in traffic flow, which deeply affects traffic efficiency and even brings security risks. 
To address this issue, this paper develops a model-based learning algorithm to optimize the disturbance attenuation level and derive an $H_{\infty}$ optimal controller for the CAVs from the nonlinear mixed traffic dynamics. %, enabling the CAVs to dampen undesired perturbations in nonlinear mixed traffic flow. 
Precisely, the main contributions of this paper are as follows: 
\begin{itemize}
    \item An affine nonlinear model is established for the mixed traffic system based on the LCC framework. Compared with existing work where linearized dynamics around equilibrium states are under consideration~\cite{cui2017stabilizing,jin2017optimal,wang2022leading}, we directly focus on the nonlinear dynamics to design model-based learning control policies for the CAVs.
    \item The $H_{\infty}$ control problem of the nonlinear mixed traffic system is formulated as a zero-sum game, whose control policy can be obtained by solving the converted Hamilton-Jacobi (HJ) inequality. The obtained state-feedback controller is proved to achieve the given attenuation level for the mixed traffic system. 
    \item The HJ inequality reserves the optimization space for attenuation level. Accordingly, we further develop a model-based learning algorithm, which optimizes the attenuation performance in outer-loop iterations through sum of squares programs, and generates stabilizing controllers with attenuation performance guarantees at every inner-loop iteration in a PI paradigm. 
\end{itemize}

The rest of this paper is organized as follows. Section~\ref{sec.problem_statement} establishes the nonlinear mixed traffic model. The model-based learning algorithm and its theoretical analysis are presented in Section~\ref{sec.method}. Section~\ref{sec.simulation} shows the simulation results, and Section~\ref{sec.conclusion} concludes this work.

%The rest of this paper is organized as follows. Section~\ref{sec.problem_statement} establishes the nonlinear model of the mixed traffic system. A PI method is presented in Section~\ref{sec.method} to solve the nonlinear $H_{\infty}$ optimal control problem of the mixed traffic system and derive the corresponding robust controller. Section~\ref{sec.simulation} shows the simulation results of the developed algorithm and verify its disturbance attenuation performance. Section~\ref{sec.conclusion} concludes this work. 

%%%%%%%%%%%%%%%%%%%%%%%%%%%%%%%%%%%%%%%%%%%%%%%%%%%%%%%%%%%%%%%%%%%%%%%%%%%%%%%%
\section{Nonlinear Modeling of Mixed Traffic}
\label{sec.problem_statement}

%In this section, we present the nonlinear model in the LCC framework and the problem statement of disturbance attenuation.

%\subsection{Nonlinear Modelling for HDVs}
%\label{subsec.driver_model}

Consider the mixed traffic system shown in Fig.~\ref{fig:mixed_platoon}, where there is one head vehicle (indexed as $0$), $m$ CAVs and $n-m$ HDVs following behind (indexed from $1$ to $n$). Without loss of generality, we assume that the first vehicle behind the head vehicle is CAV. Denote $S=\{l_1,l_2,\ldots,l_m\}$ as the set of all the CAV indexes, where $1=l_1<l_2<\cdots<l_m\leq n$. Such a multi-vehicle system in mixed traffic is named as a special form of LCC~\cite{wang2022leading}, which incorporates both upstream and downstream traffic information, and allows the CAV to attenuate the disturbances from the head vehicle, whilst actively leading the motion of the HDVs behind. %Without loss of generality, this paper assumes that the HDV immediately in front of the CAV is the head vehicle. 

Denote $p_i(t)$ and $v_i(t)$ as the position and velocity of vehicle $i$, respectively. Then, $s_i(t) = p_{i-1}(t) - p_i(t)$ and ${\dot{s}}_i(t) = v_{i-1}(t) - v_i(t)$ represent the spacing (relative distance) and relative velocity of vehicle $i$ with respect to its predecessor. Motivated by recent research on mixed traffic~\cite{jin2017optimal,wang2020controllability,molnar2020open}, we consider the OVM model for the HDVs, a typical nonlinear car-following model, to represent its longitudinal driving behavior, which is given by~\cite{bando1995dynamical}
\begin{equation} 
    \label{eq.ovm_model}
    \dot{v_i}(t) = \alpha_i \left(v^\mathrm{d}(s_i(t)) - v_i(t)\right) +\beta_i {\dot{s}}_i(t),\; i \notin S,
\end{equation}
where $\alpha_i$, $\beta_i$ denote the sensitivity coefficients for vehicle $i$, and the spacing-dependent desired velocity $v^\mathrm{d}(s_{i}(t))$ is
\begin{equation} 
    \label{eq.piecewise}
    v^\mathrm{d}( s_{i}(t) ) =
        \begin{cases}
            0, &s_{i}(t) \leq s_\mathrm{st}, \\
            \bar{v}^\mathrm{d}\left( s_{i}(t) \right), &s_\mathrm{st} < s_{i}(t) < s_\mathrm{go}, \\
            v_{\rm max}, &s_{i}(t) \geq s_\mathrm{go}, 
        \end{cases}
\end{equation}
with $\bar{v}^\mathrm{d}$ given by
\begin{equation} 
    \nonumber
    \bar{v}^\mathrm{d}\left( s_{i}(t) \right) = \frac{v_{\rm max}}{2}\left( {1 - {\cos\left( {\frac{s_{i}(t) - s_\mathrm{st}}{s_\mathrm{go} - s_\mathrm{st}}\pi} \right)}} \right).
\end{equation}
%with $v_{\rm max} = 30 ~ {\rm m/s}$, $s_{st} = 5 ~ {\rm m}$, and $s_{go} = 35 ~ {\rm m}$. 

\begin{comment}
Based on OVM model \eqref{eq.ovm_model}, the parameters of equilibrium state can be derived. Without considering the heterogeneity between vehicles, there are desired car-following distance $s^*$ and desired speed $v^*$ at the equilibrium, and the error of the desired car-following speed $\dot{s}^* = 0$. At the equilibrium state, the following equation should be satisfied: 
\begin{equation} 
    \nonumber
    F_i\left( {s^*, 0, v^*} \right) = \alpha_i \left(V(s^*) - v^*\right) = 0. 
\end{equation}
Set that $0 < v^* < v_{\rm max}$. It can be derived that $f_v(s^*) = v^*$, and 
\begin{equation} 
    \nonumber
    s^{*} = \frac{s_{go} - s_{st}}{\pi} {\rm arccos} \left( {1 - \frac{2v^{*}}{v_{\rm max}}} \right) + s_{st}. 
\end{equation}
Substituting $v^* = 15 ~ {\rm m/s}$ into the above equation yields $s^* = 20 ~ {\rm m}$. 
\end{comment}

Define the deviation state of each vehicle from the equilibrium state as
$
        {\tilde{s}}_{i}(t) = s_{i}(t) - s^{*},\,
        {\tilde{v}}_{i}(t) = v_{i}(t) - v^{*}
$, 
where $s^*$, $v^*$ denote the equilibrium spacing and velocity, respectively. For simplicity, a homogeneous setup for $s^*$ is under consideration, but all the results can be generalized to the heterogeneous case. 
Denote
% \begin{equation}
%     \nonumber
%     x_{i} = \begin{bmatrix}
%         {\tilde{s}}_{i} \\
%         {\tilde{v}}_{i} \\
%     \end{bmatrix}. 
% \end{equation}
$
    x_{i}(t) = \left[{\tilde{s}}_{i}(t), {\tilde{v}}_{i}(t)\right]^{\top} 
$
as the states of vehicle $i$. 
Then, the nonlinear dynamics model for each HDV around the equilibrium state is obtained as follows
\begin{equation} 
    \label{eq.dynamic_hdv}
    \begin{aligned} 
        {\dot{\tilde{s}}}_{i}(t) &= {\tilde{v}}_{i-1}(t) - {\tilde{v}}_{i}(t), \\
        % {\dot{\tilde{v}}}_{i}(t) &= \alpha_{i}\left( {V( {s_{i}(t)} ) - v_{i}(t)} \right) + \beta_{i}{\dot{s}}_{i}(t) \\
        {\dot{\tilde{v}}}_{i}(t) &= h_i\left({\tilde{s}}_{i}(t),\:{\tilde{v}}_{i}(t),\:{\tilde{v}}_{i-1}(t)\right), \\
    \end{aligned}
\end{equation}
where 
\begin{equation}
\begin{aligned}
        h_i(\cdot)=\;& \alpha_{i} \left( {v^\mathrm{d}( {{\tilde{s}}_{i}(t) + s^{*}} ) - \left( {{\tilde{v}}_{i}(t) + v^{*}} \right)} \right) \\
        & + \beta_{i}\left( {{\tilde{v}}_{i-1}(t) - {\tilde{v}}_{i}(t)} \right).
\end{aligned}
\end{equation}
%where $\alpha_i = 0.6$ and $\beta_i = 0.9$ are parameters of homogeneous HDVs. 

%\subsection{Mixed Platoon}
%\label{subsec.mixed_platoon}

%Considering the general human driver model (HDV), the controlled vehicle is numbered 1, and there are vehicles driven by human drivers in the front and rear. Except for the head vehicle, the position and speed information of other vehicles can be transmitted to the controlled vehicle through V2V communication. The speed of the head car fluctuates around the expected speed, which is regarded as the external disturbance of the mixed platoon system. On the other hand, considering that the vehicle driven by human drivers is affected by the speed fluctuation of the preceding vehicle, the speed fluctuation of the preceding vehicle can be regarded as external disturbance. At the same time, it is noticed that in the task of platoon control, too much consideration of the preceding vehicle will affect the control effect. Therefore, it is reasonable to only consider the immediately preceding vehicle when designing the controller for the connected automatic vehicle (CAV), 

For the CAV, its acceleration signal ${\dot{\tilde{v}}}_{l_i}(t)$ is regarded as the control input $u_i(t),\; i \in \mathbb{N}_1^m$,  where $\mathbb{N}_1^m$ denotes all the natural numbers within $\left[1,m\right]$. Then, the longitudinal control model of the CAVs is given by~\cite{wang2022leading}
\begin{equation} 
    \label{eq.dynamic_cav}
        {\dot{\tilde{v}}}_{l_i}(t) = u_i(t), \; i \in \mathbb{N}_1^m.
\end{equation}

\begin{figure}[!tb]
    \centering\includegraphics[width=0.49\textwidth]{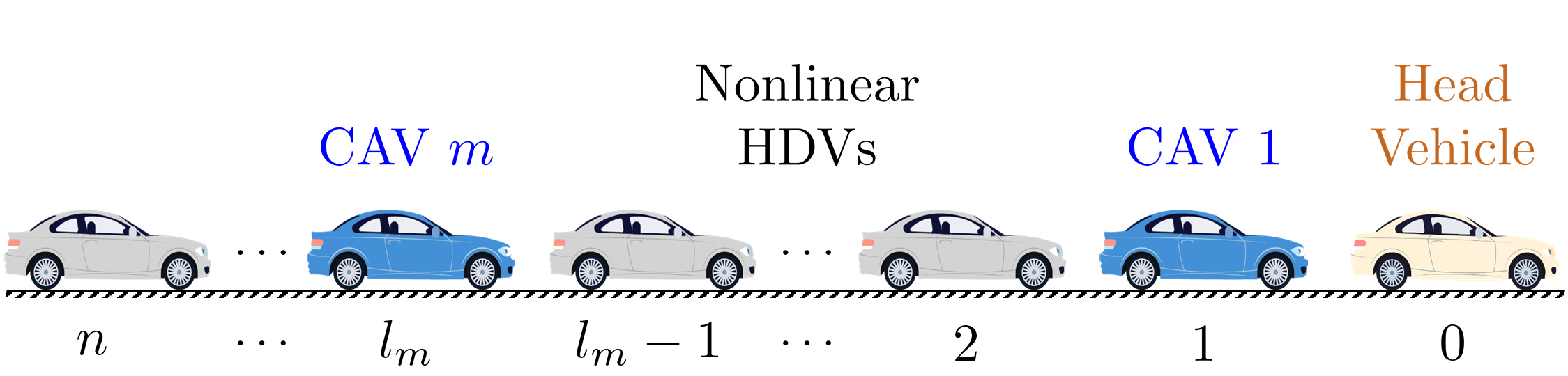}
    \vspace{-4mm}
    \caption{Schematic of the mixed traffic system in the LCC framework. There are multiple CAVs (colored in blue) and HDVs (colored in gray, with a nonlinear car-following model) following behind the head vehicle. The first vehicle behind the head vehicle is CAV.}
    \label{fig:mixed_platoon}
    \vspace{-4mm}
\end{figure}

Lumping the dynamics of the CAV and the HDVs, a nonlinear model for the LCC system can be established as follows
\begin{equation}
    \label{eq.mixed_platoon}
    \dot{x}(t) = f(x(t)) + g(x(t)) u(t) + k(x(t)) w(t), 
\end{equation}
where the lumped state and control input are defined as
\begin{equation}
    \nonumber
    \begin{aligned}
        x(t) = {\left[x_{1}^{\top}(t), x_{2}^{\top}(t), \cdots, x_{n}^{\top}(t) \right]}^{\top} \in \mathbb{R}^{2n}, \\
        u(t) = {\left[u_{1}^{\top}(t), u_{2}^{\top}(t), \cdots, u_{m}^{\top}(t) \right]}^{\top}\in \mathbb{R}^{m},
    \end{aligned}
\end{equation}
respectively, and the disturbance $w(t) = \tilde{v}_0(t) \in \mathbb{R}$ represents the velocity deviation of the head vehicle. The functions $f: \mathbb{R}^{2n} \rightarrow \mathbb{R}^{2n}$, $g: \mathbb{R}^{2n} \rightarrow  \mathbb{R}^{2n \times m}$ and $k: \mathbb{R}^{2n} \rightarrow \mathbb{R}^{2n}$ are known vector-valued functions. In the system model~\eqref{eq.mixed_platoon}, the dynamics of the rear $n - 1$ vehicles satisfy \eqref{eq.dynamic_hdv}.  %Note that a large-scale mixed traffic system can be composed of several groups of subsystems~\eqref{eq.mixed_platoon} with different scales. 
In the case of $n = 3$, $m = 1$, for example, the specific expression of the dynamic model~\eqref{eq.mixed_platoon} is as follows
\begin{equation}
    \nonumber
    f(x(t)) = \begin{bmatrix}
        {- {\tilde{v}}_{1}}(t) \\
        0 \\
        {{\tilde{v}}_{1}(t) - {\tilde{v}}_{2}}(t) \\
        % {\alpha_{2}\left\lbrack {V( {{\tilde{s}}_{2} + s^{*}} ) - \left( {{\tilde{v}}_{2} + v^{*}} \right)} \right\rbrack + \beta_{2}\left( {{\tilde{v}}_{1} - {\tilde{v}}_{2}} \right)} \\
        h\left({\tilde{s}}_{2}(t),\:{\tilde{v}}_{2}(t),\:{\tilde{v}}_{1}(t)\right) \\
        {{\tilde{v}}_{2}(t) - {\tilde{v}}_{3}}(t) \\
        % {\alpha_{3}\left\lbrack {V( {{\tilde{s}}_{3} + s^{*}} ) - \left( {{\tilde{v}}_{3} + v^{*}} \right)} \right\rbrack + \beta_{3}\left( {{\tilde{v}}_{2} - {\tilde{v}}_{3}} \right)} \\
        h\left({\tilde{s}}_{3}(t),\:{\tilde{v}}_{3}(t),\:{\tilde{v}}_{2}(t)\right) \\
    \end{bmatrix}, 
\end{equation}
% \begin{equation}
%     \nonumber
%     g(x) = \begin{bmatrix}
%         0 \\
%         1 \\
%         0 \\
%         0 \\
%         0 \\
%         0 \\
%     \end{bmatrix}, ~
%     k(x) = \begin{bmatrix}
%         1 \\
%         0 \\
%         0 \\
%         0 \\
%         0 \\
%         0 \\
%     \end{bmatrix}. 
% \end{equation}
\begin{equation}
    \nonumber
    g(x(t))\:=\:{\left[0,1,0,0,0,0\right]}^{\top}, ~ k(x(t))\:=\:{\left[1,0,0,0,0,0\right]}^{\top}.\:
\end{equation}

To describe the performance of the mixed traffic system, we define $z(t) \in \mathbb{R}^{2n + m}$ as the output 
\begin{equation}
    \label{eq.performance_output}
    z(t) \triangleq \begin{bmatrix}
    \sqrt{Q} x(t) \\
    \sqrt{R} u(t) \\
    \end{bmatrix} ,
\end{equation}
and the square of its norm is given by
\begin{equation}
    \left\| z(t) \right\|^{2} = z^{\top}(t)z(t) = x^{\top}(t)Qx(t) + u^{\top}(t)Ru(t),
\end{equation}
where $\sqrt{Q}=\mathrm{diag}(\theta_s,\theta_v,\ldots,\theta_s,\theta_v) \in \mathbb{R}^{2n \times 2n}$ and $\sqrt{R}=\mathrm{diag}(\theta_u,\ldots,\theta_u) \in \mathbb{R}^{m \times m}$ are positive definite matrices, with $\theta_s$, $\theta_v$, $\theta_u$ denoting the weight coefficients for penalizing spacing deviations, velocity deviations and control inputs. Note that the system~\eqref{eq.mixed_platoon}-\eqref{eq.performance_output} is zero-state observable. 
% When the external disturbance $w$ acts on the system, it will affect the performance output $z$. The following definition characterizes the attenuation level of the system~\eqref{eq.mixed_platoon} against the external disturbance $w$. 

\begin{remark}
Existing research mostly considers the linearized dynamics of the mixed traffic system to design the control policies for the CAVs. In this paper, we directly focus on the affine model~\eqref{eq.mixed_platoon}, and aim at developing model-based learning method with theoretical guarantees and ability to process nonlinear mixed traffic systems. Note that although the OVM model~\eqref{eq.ovm_model} is utilized for describing the HDVs' dynamics and the resulting expression of $g(x)$ is a constant function, our control method is applicable to any mixed traffic system in the general nonlinear affine form of~\eqref{eq.mixed_platoon}. 
\end{remark}

%%%%%%%%%%%%%%%%%%%%%%%%%%%%%%%%%%%%%%%%%%%%%%%%%%%%%%%%%%%%%%%%%%%%%%%%%%%%%%%%
\section{$H_{\infty}$ Optimal Control by Policy Iteration}
\label{sec.method}

This section first formulates the $H_{\infty}$ control problem of the mixed traffic system as a zero-sum game. Based on the PI framework, a model-based learning algorithm is then developed to solve the equivalent HJ inequality. 

\subsection{Problem Formulation}

To characterize the disturbance attenuation performance of the closed-loop system, we first give the following definition. For the convenience of writing, the time $t$ will be omitted in the subsequent content. 

\begin{definition}[Disturbance Attenuation]
{
For all disturbance $w \in L_{2} [0,\infty)$, the closed-loop system~\eqref{eq.mixed_platoon}-\eqref{eq.performance_output} with the initial state $x(0) = 0$ is said to have an $L_2$-gain $\leq$ $\gamma$, if 
\begin{equation}
    \nonumber
    % \label{eq.disturbance_attenuation}
    % \int_{0}^{\infty}{\left\| z \right\|^{2}dt} \leq {\gamma}^{2} \int_{0}^{\infty}{\left\| w \right\|^{2} dt}, \ \forall w \in L_{2} [0,\infty) .
    \int_{0}^{\infty}{\left\| z \right\|^{2}dt} \leq {\gamma}^{2} \int_{0}^{\infty}{\left\| w \right\|^{2} dt} .
\end{equation}
In other words, the system satisfies the disturbance attenuation performance with attenuation level $\gamma > 0$. 
}
\end{definition}

The attenuation level $\gamma$ captures the influence of the external disturbance $w$ on the performance output $z$. Precisely, a smaller value of $\gamma$ indicates a better capability of CAVs in dissipating traffic waves. Then, given the nonlinear mixed traffic system~\eqref{eq.mixed_platoon}, define the value function of the initial state $x = x(0)$ as
\begin{equation}
    \label{eq.value_function}
    V\left(x \right) \triangleq {\int_{0}^{\infty}\left( {l\left( {x(\tau),u(\tau),w(\tau)} \right)} \right)}d\tau ,
\end{equation}
where the cost function is defined as
\begin{equation}
    \label{eq.cost_function}
    l\left( {x,u,w} \right) \triangleq x^{\top}Qx + u^{\top}Ru - \gamma^{2}w^{\top}w .
\end{equation}
From the point of view of game theory, disturbance aims at deteriorating control performance, while control policy optimizes the worst-case performance in $H_{\infty}$ control~\cite{bacsar2008h}. Given a suitable attenuation level $\gamma > 0$, the $H_{\infty}$ control problem can be formulated as the following zero-sum game~\cite{lewis2012optimal}
\begin{equation}
    \label{eq.zero-sum_game}
    V^{*}(x) = \min\limits_{u(\cdot)}{\max\limits_{w(\cdot)}{{\int_{0}^{\infty}\left( {l\left( {x(\tau),u(\tau),w(\tau)} \right)} \right)}d\tau}} ,
\end{equation}
where $V^* (x)$ is the Nash value, control $u(\cdot)$ and disturbance $w(\cdot)$ are two sides of the game. Moreover, the controller at the Nash equilibrium should stabilize the system at $w \equiv 0$, and allow the closed-loop system to have an $L_2$-gain $\leq$ $\gamma$ for all $w \in L_2[0,\:\infty)$. Further, the $H_{\infty}$ optimal control problem explores the lowest attenuation level $\gamma^* > 0$ and resolves the corresponding zero-sum game~\eqref{eq.zero-sum_game}. The existence of the lowest attenuation level of nonlinear affine systems is guaranteed by~\cite{bacsar2008h}. The following assumption declares the existence of the desired controller in mixed traffic flow. 

\begin{assumption}
\label{assump.existence_of_robust_controller}
Given an attenuation level $\gamma \geq \gamma^*$, there exits a robust controller $u = \pi(x)$ with $\pi(0) = 0$ such that the system~\eqref{eq.mixed_platoon}-\eqref{eq.performance_output} is stabilized at $w \equiv 0$ and that the closed-loop system has an $L_2$-gain $\leq$ $\gamma$ for all $w \in L_2[0,\:\infty)$. 
\end{assumption}

On the premise of zero-state observability, the solution to the following Hamilton--Jacobi--Isaacs (HJI) equation solves the zero-sum game~\cite{lewis2012optimal}
\begin{equation}
    \label{eq.HJI}
    \begin{aligned}
        x^{\top}Qx &+ (\nabla V^{*}(x))^{\top}f(x) \\
        &- \frac{1}{4}(\nabla V^{*}(x))^{\top} g(x)R^{- 1}g^{\top}(x) \nabla V^{*}(x) \\
        &+ \frac{1}{4\gamma^{2}}(\nabla V^{*}(x))^{\top} k(x)k^{\top}(x) \nabla V^{*}(x) = 0 ,
    \end{aligned}
\end{equation}
which is a nonlinear partial differential equation about the optimal value function $V^* (x)$ with the boundary condition $V^* (0) = 0$. If the HJI equation has a smooth positive semi-definite solution $V^*(x)$, then the controller is derived as 
\begin{equation}
    \nonumber
    u^{*}(x) = - \frac{1}{2}R^{- 1}g^{\top}(x)\nabla V^{*}(x) .
\end{equation}

Traditional PI algorithms usually focus on providing numerical methods for solving the HJI equation~\eqref{eq.HJI} to generate robust controllers. Through the following lemma, or as shown in Theorem~\ref{thm.stability_and_attenuation_perf}, we can also derive a robust controller from the associated HJ inequality~\eqref{eq.HJ_inequality}. Compared with solving HJI equation directly, the gap of HJ inequality allows for further optimizing the attenuation level.

\begin{lemma}[\!\!{\cite[Theorem 16 \& Corollary 17]{van1992}}]
\label{lemma.HJ_equation_and_inequality}
Consider the nonlinear system~\eqref{eq.mixed_platoon}-\eqref{eq.performance_output} with an attenuation level $\gamma$. Suppose that there is a smooth positive semi-definite solution $V(x)$ to the HJI equation~\eqref{eq.HJI} or the HJ inequality 
\begin{equation}
    \label{eq.HJ_inequality}
    \begin{aligned}
        x^{\top}Qx &+ (\nabla V(x))^{\top}f(x) \\
        &- \frac{1}{4}(\nabla V(x))^{\top} g(x)R^{- 1}g^{\top}(x) \nabla V(x) \\
        &+ \frac{1}{4\gamma^{2}}(\nabla V(x))^{\top} k(x)k^{\top}(x) \nabla V(x) \leq 0 ,
    \end{aligned}
\end{equation}
with the boundary condition $V(0) = 0$, then the closed-loop system with the state feedback controller 
\begin{equation}
    \nonumber
    u(x) = - \frac{1}{2}R^{- 1}g^{\top}(x)\nabla V(x) ,
\end{equation}
is asymptotically stable at $w \equiv 0$, and has an $L_2$-gain $\leq$ $\gamma$ for all disturbance $w \in L_{2} [0,\infty)$. 
\end{lemma}

Besides, the existence of the solution to the HJ inequality is guaranteed by the following lemma. 
\begin{lemma}[{\cite[Theorem 18]{van1992}}]
\label{lemma.converse}
Consider the nonlinear system~\eqref{eq.mixed_platoon}-\eqref{eq.performance_output} and an attenuation level $\gamma$. If there is a controller $u = \pi(x)$ satisfying Assumption~\ref{assump.existence_of_robust_controller}, there exists a smooth positive semi-definite solution $V_a(x)$ to the HJ inequality~\eqref{eq.HJ_inequality}. 
\end{lemma}

Similar to the HJI equation~\eqref{eq.HJI}, the HJ inequality~\eqref{eq.HJ_inequality} contains two nonlinear terms about the differential of value function, which are related to control input function~$g(x)$ and disturbance input function~$k(x)$, respectively. This makes it non-trivial to get the problem solutions.

\subsection{Inequality Conversion}

We proceed to provide a concrete procedure to solve the HJ inequality~\eqref{eq.HJ_inequality}. To begin with, the inequality is transformed to eliminate the nonlinear differential term about disturbance input function~$k(x)$ while retaining the characteristics of the control policy. 

% By completing the square in the HJ inequality~\eqref{eq.HJ_inequality}, we have the following inequality for all disturbance signal $w$: 
% \begin{equation}
%     \nonumber
%     \begin{aligned}
%         &~~~~x^{\top}Qx + (\nabla V(x))^{\top}f(x) \\
%         &~~~~~~~~~~~ - \frac{1}{4}(\nabla V(x))^{\top} g(x)R^{- 1}g^{\top}(x) \nabla V(x) \\
%         &~~~~~~~~~~~ + (\nabla V(x))^{\top} k(x) w - \gamma^2 w^{\top} w \\
%         &\leq -\frac{1}{4\gamma^{2}}(\nabla V(x))^{\top} k(x)k^{\top}(x) \nabla V(x) \\
%         &~~~~~~~~~~~ + (\nabla V(x))^{\top} k(x) w - \gamma^2 w^{\top} w \\
%         &= - \gamma^2 \left\| w - \frac{1}{2\gamma^2}k^{\top}(x)\nabla V(x)  \right\|^2 \\
%         &\leq 0. 
%     \end{aligned}
% \end{equation}
In order to facilitate the solution through conversion, add the following square term
\begin{equation}
    \nonumber
    - \gamma^2 \left\| w - \frac{1}{2\gamma^2}k^{\top}(x)\nabla V(x)  \right\|^2 \leq 0, 
\end{equation}
to both sides of the HJ inequality~\eqref{eq.HJ_inequality} and get the following converted inequality for all disturbance signal $w$
\begin{equation}
    \label{eq.converted_inequality}
    \begin{aligned}
        x^{\top}Qx &+ (\nabla V(x))^{\top}f(x) \\
        &- \frac{1}{4}(\nabla V(x))^{\top} g(x)R^{- 1}g^{\top}(x) \nabla V(x) \\
        &+ (\nabla V(x))^{\top} k(x) w - \gamma^2 w^{\top} w \leq 0. 
    \end{aligned}
\end{equation}
% Note that the above inequality is held for arbitrary disturbance, and thus $w$ is regarded as an independent variable. 

This transformed inequality~\eqref{eq.converted_inequality} is exactly the problem that we aim to solve in this work. It can be proved that the controller derived from the feasible solution of the converted inequality~\eqref{eq.converted_inequality} reserves the stability and disturbance attenuation performance. With Lemma~\ref{lemma.converse} in place, it follows that the inequality~\eqref{eq.converted_inequality} admits a feasible solution $V_a(x)$. 
% % relaxed $L_2$-gain optimization problem
% \begin{definition}[relaxed $L_2$-gain optimization problem]
% \textnormal{
% The relaxed $L_2$-gain optimization problem is defined as  
% \begin{equation}
%     \nonumber
%     \min_{V}{ \int_{\Omega}{V(x) dx} }
% \end{equation}
% \begin{equation}
%     \label{eq.negative}
%     \begin{aligned}
%         &s.t.~~ -x^{\top}Qx - (\nabla V(x))^{\top}f(x) \\
%         &~~~~~~~~~~~~~~~~~ + \frac{1}{4}(\nabla V(x))^{\top} g(x)R^{- 1}g^{\top}(x) \nabla V(x) \\
%         &~~~~~~~~~~~~~~~~~ - (\nabla V(x))^{\top} k(x) w + \gamma^2 w^{\top} w \geq 0 
%     \end{aligned}
% \end{equation}
% \begin{equation}
%     \nonumber
%     V(x) \geq 0, ~~~~~~~~~~~~~~~~~~~~~~~~~~~~~~~~~~~~~~
% \end{equation}
% where $\Omega \subseteq \mathbb{R}^{2n}$ is a compact set containing the origin. 
% }
% \end{definition}

\begin{theorem}[Stability and Robustness] %and Attenuation Performance]
\label{thm.stability_and_attenuation_perf}
Suppose that the converted inequality~\eqref{eq.converted_inequality} admits a feasible solution $V(x)$. Then, the closed-loop system with the controller 
\begin{equation}
    \nonumber
    u(x) = - \frac{1}{2}R^{- 1}g^{\top}(x)\nabla V(x) ,
\end{equation}
is asymptotically stable at $w \equiv 0$, and has an $L_2$-gain $\leq$ $\gamma$ for all $w \in L_{2} [0,\infty)$. 
\end{theorem}

\begin{proof}
Substituting the expression of $u(x)$ into the transformed inequality~\eqref{eq.converted_inequality} yields 
\begin{equation}
    \label{eq.inequality_Hamilton}
    \begin{aligned}
        &(\nabla V(x))^{\top} \big( f(x) + g(x)u(x) + k(x)w \big) \\
        \leq & -x^{\top}Qx - u^{\top}(x)Ru(x) + \gamma^{2}w^{\top}w .
    \end{aligned}
\end{equation}
When $w \equiv 0$, one has 
\begin{equation}
    \nonumber
    (\nabla V(x))^{\top} \big( f(x) + g(x)u(x) \big) \leq -x^{\top}Qx - u^{\top}(x)Ru(x) \leq 0. 
\end{equation}
Therefore, the asymptotic stability is obtained by Lyapunov's direct method, where $V(x)$ is a Lyapunov function candidate. For all $w \in L_{2} [0,\infty)$, by integrating the derived inequality~\eqref{eq.inequality_Hamilton}, it can be directly obtained by~\cite[Theorem 16]{van1992} that the closed-loop system has an $L_2$-gain $\leq$ $\gamma$. 
%Moreover, for arbitrary initial state $x(0)$, the performance output is also in $L_{2} [0,\infty)$. 
\end{proof}

\subsection{Model-based Learning Algorithm}

With stability and robustness results shown in Theorem~\ref{thm.stability_and_attenuation_perf}, we are ready to design a model-based learning algorithm to solve the converted inequality~\eqref{eq.converted_inequality}. Precisely, given a desired attenuation level $\gamma$, the inner-loop iteration of the algorithm employs the policy iteration method to derive stabilizing controllers. When the iterative process converges, the converted inequality~\eqref{eq.converted_inequality} can be restored by substituting the improved control policy~\eqref{eq.policy_improvement} into the Hamiltonian constraint~\eqref{eq.Hamiltonian_constraint}. In outer-loop iteration, the attenuation level is optimized by using the gap of the Hamiltonian constraint. 
% The schematic of the developed method is shown in Fig~, and the pseudocode is shown in Algorithm~\ref{alg}. 
The pseudocode of the developed method is shown in Algorithm~\ref{alg}. In the following, we present further elaborations and theoretical guarantees on the developed algorithm. %The following subsections will introduce the details of the presented method.

\begin{algorithm}[!tb]
    \caption{Model-based Learning Algorithm}
    \label{alg}
    \LinesNumbered 
    \KwIn{initial control policy $u^{(0)}(x)$.}
    \For{$i = 1, 2, \cdots$}{
        \textbf{Attenuation Level Optimization:}
        % Optimize attenuation level $\gamma^{(i)}$ by: 
        \begin{subequations}
            \label{eq.gamma_optimization}
            \begin{align}
                \gamma^{(i)} = \text{arg} \min_{\gamma>0}{ \gamma } \\
                \mathrm{s.t.} ~~ \mathcal{L} \left( V(x), u^{(i-1)}(x), \gamma \right) &\geq 0 \\
                V(x) &\geq 0 .
            \end{align}
        \end{subequations}
        
        Let $V_0^{(i)}(x) \gets V(x)$ and $u_0^{(i)}(x) \gets u^{(i-1)}(x)$. 
        
        \For{$k = 1, 2, \cdots$}{
            \textbf{Policy Evaluation:}
            % Solve value function $V_k^{(i)}(x)$ by: 
            \begin{subequations}
                \label{eq.policy_evaluation}
                \begin{align}
                    \label{eq.optimization_objective}
                    V_k^{(i)}(x) = \text{arg} \min_{V}{ \int_{\Omega}{V(x) dx} } \\
                    \label{eq.Hamiltonian_constraint}
                    \mathrm{s.t.} ~~ \mathcal{L} \left( V(x), u_{k-1}^{(i)}(x), \gamma^{(i)} \right) &\geq 0 \\
                    \label{eq.non-increasing_value}
                    V_{k-1}^{(i)}(x) - V(x) &\geq 0 \\
                    \label{eq.semi-positive_definite}
                    V(x) &\geq 0 .
                \end{align}
            \end{subequations}
            
            \textbf{Policy Improvement:}
            % Update control policy $u_{k}^{(i)}(x)$ by: 
            \begin{equation}
                \label{eq.policy_improvement}
                u_{k}^{(i)}(x) = - \frac{1}{2}R^{- 1}g^{\top}(x)\nabla V_k^{(i)}(x) .
            \end{equation}
        }
        Let $u^{(i)}(x) \gets u_{\infty}^{(i)}(x)$.
    }
\end{algorithm}

\vspace{1mm}
\noindent \textbf{(Inner-loop) Policy Iteration:}
% Enlightened by the existing PI framework~\cite{sutton2018reinforcement}, the nonlinear differential term about control input function~$g(x)$ can be removed by iterating between policy evaluation~\eqref{eq.policy_evaluation} and policy improvement~\eqref{eq.policy_improvement}. 
For an attenuation level $\gamma^{(i)} \geq \gamma^*$, a stabilizing controller is designed in the inner-loop iteration to allow the closed-loop system to have an $L_2$-gain smaller than $\gamma^{(i)}$. 
Enlightened by the existing PI framework~\cite{li2022reinforcement}, the step of policy improvement~\eqref{eq.policy_improvement} allows the nonlinear differential term about control input function~$g(x)$ to be simplified to a linear term~\eqref{eq.Hamiltonian_constraint} in the step of policy evaluation~\eqref{eq.optimization_objective}. Consider the negative Hamiltonian as 
\begin{equation}
    \label{eq.inverse_Hamiltonian}
    \begin{aligned}
        \mathcal{L} \left( V(x), u(x), \gamma \right) &\triangleq - (\nabla V(x))^{\top} \big( f(x)\!+\!g(x)u(x)\!+\!k(x)w \big) \\
        &~~~-x^{\top}Qx - u^{\top}(x)Ru(x) + \gamma^{2}w^{\top}w .
    \end{aligned}
\end{equation}

Given an improved controller $u_{k-1}^{(i)}(x)$ at the beginning of the $k$-th iteration, the value function $V_k^{(i)}(x)$ is updated by imposing a constraint on Hamiltonian~\eqref{eq.Hamiltonian_constraint} in policy evaluation, which only contains the linear term of the differential of value function. The following lemma illustrates the existence of the initial feasible solution of the PI framework. 

% \begin{algorithm}[!htb]
% \caption{Policy Iteration for Relaxed $L_2$-gain Optimization Problem}
% \label{alg:policy_iteration_for_relaxed}
% \begin{algorithmic}
%     \STATE Input: a given attenuation level $\gamma$, an initial value function $V_0(x)$, an initial control policy $u_1(x)$
%     \FOR{$i=1, 2, \cdots$}
%         \STATE \textbf{Policy Evaluation:}
%         \STATE Solve value function $V_i(x)$ by: 
%         \begin{equation}
%             \label{eq.policy_evaluation}
%             \min_{V}{ \int_{\Omega}{V(x) dx} } 
%         \end{equation}
%         \begin{equation}
%             \nonumber
%             \begin{aligned}
%                 s.t. ~~ \mathcal{L} \left( V(x), u_i(x), \gamma \right) &\geq 0 \\
%                 V_{i-1}(x) - V(x) &\geq 0 \\
%                 V(x) &\geq 0 . \\
%             \end{aligned}
%         \end{equation}
        
%         \STATE \textbf{Policy Improvement:}
%         \STATE Update control policy $u_{i+1}(x)$ by: 
%         \begin{equation}
%             \label{eq.policy_improvement}
%             u_{i+1}(x) = - \frac{1}{2}R^{- 1}g^{\top}(x)\nabla V_i(x) .
%         \end{equation}
%     \ENDFOR
% \end{algorithmic}
% \end{algorithm}

\begin{lemma}[Feasibility of Hamiltonian Constraint]
\label{lemma.existence_of_value}
There exists a controller $u_a(x)$ such that the Hamiltonian constraint $\mathcal{L} \left( V(x), u_a(x), \gamma^{(i)} \right) \geq 0$ has a non-empty feasible set about the value function $V(x)$.
\end{lemma}

\begin{proof}
According to Lemma~\ref{lemma.converse}, the value function $V_a(x)$ satisfies the HJ inequality~\eqref{eq.HJ_inequality}. Construct the controller as $u_a(x) = - \frac{1}{2}R^{- 1}g^{\top}(x)\nabla V_a(x)$. It is straightforward that $\mathcal{L} \left( V_a(x), u_a(x), \gamma^{(i)} \right) \geq 0$. So, $V_a(x)$ is a feasible solution to the Hamiltonian constraint $\mathcal{L} \left( V(x), u_a(x), \gamma^{(i)} \right) \geq 0$. 
\end{proof}

Therefore, the value function and control policy can be initialized as $V_0^{(i)}(x) = V_a(x)$ and $u_0^{(i)}(x) = u_a(x)$ such that the first iteration of policy evaluation has a feasible solution. Besides, it can be proved recursively that the subsequent iterations of policy evaluation~\eqref{eq.policy_evaluation} have a feasible solution. 

\begin{theorem}[Recursive Feasibility]
\label{thm.feasibility}
Consider the PI paradigm with the policy evaluation~\eqref{eq.policy_evaluation} and the policy improvement~\eqref{eq.policy_improvement}. If policy evaluation is feasible at the $k$-th iteration, it will also be feasible at the $(k+1)$-th iteration. 
\end{theorem}

\begin{proof}
Assume that for $u_{k-1}^{(i)}(x)$, the Hamiltonian constraint~\eqref{eq.Hamiltonian_constraint} in policy evaluation has a feasible solution $V_k^{(i)}(x)$, \emph{i.e.}, $\mathcal{L} \left( V_k^{(i)}(x), u_{k-1}^{(i)}(x), \gamma^{(i)} \right) \geq 0$. After updating the control policy $u_{k}^{(i)}(x)$ at the policy improvement step~\eqref{eq.policy_improvement}, we have 
\begin{equation}
    \label{eq.induction}
    \begin{aligned}
        &~\mathcal{L} \left( V_k^{(i)}(x), u_{k}^{(i)}(x), \gamma^{(i)} \right) \\
        % = &- (\nabla V_k^{(i)}(x))^{\top} \big( f(x) + g(x)u_{k}^{(i)}(x) + k(x)w \big) \\
        % &-x^{\top}Qx - {u_{k}^{(i)}}^{\top}(x)Ru_{k}^{(i)}(x) + {\gamma^{(i)}}^{2}w^{\top} w \\
        % = &- (\nabla V_k^{(i)}(x))^{\top} \big( f(x) + g(x)u_{k-1}^{(i)}(x) + k(x)w \big) \\
        % &- (\nabla V_k^{(i)}(x))^{\top} g(x) (u_{k}^{(i)}(x) - u_{k-1}^{(i)}(x)) - {u_{k-1}^{(i)}}^{\top}(x)Ru_{k-1}^{(i)}(x) \\
        % &+ {u_{k-1}^{(i)}}^{\top}(x)Ru_{k-1}^{(i)}(x) -x^{\top}Qx - {u_{k}^{(i)}}^{\top}(x)Ru_{k}^{(i)}(x) + \gamma^{2}w^{\top}w \\
        % = &~\mathcal{L} \left( V_k^{(i)}(x), u_{k-1}^{(i)}(x), \gamma \right) - (\nabla V_k^{(i)}(x))^{\top} g(x) (u_{k}^{(i)}(x) - u_{k-1}^{(i)}(x)) \\
        % &+ {u_{k-1}^{(i)}}^{\top}(x)Ru_{k-1}^{(i)}(x) - {u_{k}^{(i)}}^{\top}(x)Ru_{k}^{(i)}(x) \\
        % = &~\mathcal{L} \left( V_k^{(i)}(x), u_{k-1}^{(i)}(x), \gamma^{(i)} \right) + 2{u_{k}^{(i)}}^{\top}(x) R (u_{k}^{(i)}(x) - u_{k-1}^{(i)}(x)) \\
        % &+ {u_{k-1}^{(i)}}^{\top}(x)Ru_{k-1}^{(i)}(x) - {u_{k}^{(i)}}^{\top}(x)Ru_{k}^{(i)}(x) \\
        = &~\mathcal{L} \left( V_k^{(i)}(x), u_{k-1}^{(i)}(x), \gamma^{(i)} \right) \\
        &+ \left(u_{k}^{(i)}(x) - u_{k-1}^{(i)}(x)\right)^{\top} R \left(u_{k}^{(i)}(x) - u_{k-1}^{(i)}(x)\right) \geq 0.
        % \geq &~\mathcal{L} \left( V_k^{(i)}(x), u_{k-1}^{(i)}(x), \gamma^{(i)} \right), 
    \end{aligned}
\end{equation}
Therefore, $V_{k+1}^{(i)}(x) = V_k^{(i)}(x)$ is at least a feasible solution of the policy evaluation at the $(k+1)$-th iteration. 
% The remaining iteration steps can be proved by mathematical induction. 
\end{proof}

%Besides, according to Theorem~\ref{thm.stability_and_attenuation_perf}, the initial control policy $u_a(x)$ makes the closed-loop system have an $L_2$-gain $\leq$ $\gamma$. 

\begin{corollary}[Recursive Stability and Robustness]% and Attenuation Performance]
\label{corollary.recursive_stability}
If the initial control policy makes the Hamiltonian constraint~\eqref{eq.Hamiltonian_constraint} feasible, the closed-loop system with the control policy~\eqref{eq.policy_improvement} at every iteration step is asymptotically stable and has an $L_2$-gain $\leq$ $\gamma^{(i)}$. 
\end{corollary}

Note that, at every inner-loop iteration, substituting the improved control policy~\eqref{eq.policy_improvement} into the Hamiltonian constraint~\eqref{eq.Hamiltonian_constraint} yields the converted inequality~\eqref{eq.converted_inequality}. Thus, Corollary~\ref{corollary.recursive_stability} can be directly obtained from Theorem~\ref{thm.stability_and_attenuation_perf}. 
The aforementioned analysis establishes the stability and disturbance attenuation performance of the controller during the implementation of the PI process in the inner loop of Algorithm~\ref{alg}. 
% It is worth noting that when the algorithm terminates at any number of iteration steps, the derived controller always has desired closed-loop performance. 
Because the controller generated by each iteration of the inner loop has desired performance, the setting of termination condition has become an open problem. 

% So far, we have found the overall process stability and disturbance attenuation performance of the controller during the implementation of the PI algorithm. It is worth noting that when the algorithm terminates at any number of iteration steps, the derived controller always has desired closed-loop performance. Therefore, the setting of termination condition has become an open problem. In order to apply convergence guidance to the value function, the last two inequalities is imposed in policy evaluation~\eqref{eq.policy_evaluation} to ensure that the value function is monotonically non-increasing and semi-positive definite, respectively. Besides, the above conditions are all constraints imposed on the value function. In order to construct an optimization problem in policy evaluation, the integral of value function in interested state space $\Omega$ is selected as the optimization objective~\cite{jiang2015global, zhu2017policy}, where $\Omega \subseteq \mathbb{R}^{2n}$ is a compact set containing the origin. 

In order to apply convergence guidance to the value function, the two inequalities~\eqref{eq.non-increasing_value} and~\eqref{eq.semi-positive_definite} are imposed during the policy evaluation step to ensure that the value function is monotonically non-increasing and semi-positive definite, respectively. To formulate an optimization problem in policy evaluation, the integral of value function in interested state space $\Omega \subseteq \mathbb{R}^{2n}$ is selected as the optimization objective~\eqref{eq.optimization_objective}, where $\Omega$ is a compact set~\cite{zhu2017policy}. 

\vspace{1mm}
\noindent \textbf{(Outer-loop) Attenuation Level Optimization:}
After policy improvement step, the gap hidden in the Hamiltonian constraint~\eqref{eq.induction} implies that it is possible to find a smaller attenuation level $\gamma^{(i)} \leq \gamma^{(i-1)}$ to ensure that 
\begin{equation}
    \nonumber
    \mathcal{L} \left( V_{k}^{(i-1)}(x), u_{k}^{(i-1)}(x), \gamma^{(i)} \right) \geq 0. 
\end{equation}
Therefore, the obtained controller $u_{k}^{(i-1)}(x)$ has the potential to achieve a smaller $L_2$-gain for the closed-loop system. A smaller attenuation level can be found by solving the sum of squares program~\eqref{eq.gamma_optimization}, which has a non-empty feasible set that contains at least $\gamma^{(i-1)}$. By wrapping the optimization of attenuation level in the outer loop of the PI framework, a numerical method for approximating the solution to the $H_{\infty}$ optimal control problem is obtained. 

% \subsection{Practical Implementation}

% In the implementation of the algorithm, the value function is parameterized as
% \begin{equation}
%     \nonumber
%     V(x) = \sum_{i=1}^{m} {c_i\varphi_i(x)} ,
% \end{equation}
% where $\{\varphi_i(x)\}_{i=1}^m$ is a set of basis functions, such as polynomials, and $\{c_i\}_{i=1}^m$ are the parameters to be optimized. Then, the attenuation level optimization~\eqref{eq.gamma_optimization} and the policy evaluation step~\eqref{eq.policy_evaluation} are constructed as sum of squares programs (SOSPs)~\cite{zhu2017policy}. They can be reformulated as semi-definite programs (SDPs)~\cite{papachristodoulou2013sostools}, which can be solved efficiently by interior point methods provided by free MATLAB solvers, such as SeDuMi. 
% % In this case, SOSTOOLS \cite{papachristodoulou2013sostools} provides an automated process for the problem transformation from SOSP to SDP, and reverts the SDP solution back to the original form of SOSP. 
% In this case, SOSTOOLS \cite{papachristodoulou2013sostools} provides an automated process for the problem transformation from SOSP to SDP and solution conversion from SDP back to SOSP, enabling users to avoid the error of manual conversion. After a period of training, the attenuation level and control policy can be output when there is almost no change. 

\begin{remark}
In the implementation of the algorithm, the value function is parameterized as
\begin{equation}
    \label{eq.parameterized_value}
    V(x) = \sum_{i=1}^{m} {c_i\varphi_i(x)} ,
\end{equation}
where $\{\varphi_i(x)\}_{i=1}^m$ is a set of basis functions, such as polynomials, and $\{c_i\}_{i=1}^m$ are the parameters to be optimized. The attenuation level optimization~\eqref{eq.gamma_optimization} and the policy evaluation step~\eqref{eq.policy_evaluation} are constructed as sum of squares programs~\cite{zhu2017policy}, which can be conveniently solved via SOSTOOLS.% \cite{papachristodoulou2013sostools}. 
\end{remark}

%%%%%%%%%%%%%%%%%%%%%%%%%%%%%%%%%%%%%%%%%%%%%%%%%%%%%%%%%%%%%%%%%%%%%%%%%%%%%%%%
\section{Traffic Simulations}
\label{sec.simulation}

In this section, we present the nonlinear traffic simulations and analyze the performance of the developed model-based learning control policy. The nonlinear OVM model with a typical parameter setup~\cite{jin2017optimal} is employed for the HDVs.% the developed PI algorithm is applied to the mixed traffic system to verify its control effect. 
% In order to compare the effect of centralized controller and distributed controller, the algorithm can be implemented on different combinations of platoon subsystems~\eqref{eq.mixed_platoon} with different orders of polynomial value functions. 

For the parameter setup in the controller, the weight coefficients are set as $\theta_s=0.03$, $\theta_v=0.15$, $\theta_u=1$.
\begin{comment}
    \begin{equation}
    \nonumber
    z = {\left[0.03 s_1, 0.15 v_1, 0.03 s_2, 0.15 v_2, 0.03 s_3, 0.15 v_3, u\right]}^{\top}. 
\end{equation}
\end{comment}
To conveniently solve sum of squares programs via SOSTOOLS toolbox, the piecewise function~\eqref{eq.piecewise} is approximated by a quintic polynomial. %, i.e., 
% \begin{equation}
%     \nonumber
%     \begin{aligned}
%         V(s_i + s^*) &\approx 0.000005126368 \times s_i^5 - 0.004058374 \times s_i^3 \\ 
%         &~~~+ 1.653612 \times s_i + 15.0. 
%     \end{aligned}
% \end{equation}
In Algorithm~\ref{alg}, a quartic polynomial consisting of $144$ terms is employed to approximate the value function by~\eqref{eq.parameterized_value}.
The set $\Omega$ is in the range of $|s_i| \leq 4$ and $|v_i| \leq 5$. The initial controller is chosen as $u^{(0)} = 0.5 s_1 - 1.0 v_1$. 
In the simulations, a sinusoidal disturbance signal $w(t) = 5 \text{sin} (20t / \pi )\;\mathrm{m/s}$ is imposed on the head vehicle. % in front of the traffic flow. 
We first consider a small-scale mixed traffic system with three following vehicles, \emph{i.e.}, $n=3$.
As shown in Fig.~\ref{subfig:velocity_all_hdv}, when all the vehicles are HDVs, the velocity oscillations persist during its propagation. For comparison, the velocity perturbations are apparently mitigated by the proposed controller even after one single iteration (see Fig.~\ref{subfig:velocity_initial_controller}). This result validates the effectiveness of the inner-loop PI paradigm in Algorithm~\ref{alg}: the proposed controller can stabilize the mixed traffic system at each iteration step with attenuation performance guarantees. Further, with more iterations conducted, the attenuation level can be gradually and continuously improved (see Fig.~\ref{subfig:velocity_final_controller} for the performance after $20$ iterations and Fig.~\ref{subfig:attenuation_comparison} for the attenuation level during the simulations after different iterations).  These results validate the performance of the outer-loop attenuation level optimization in Algorithm~\ref{alg}.
%As expected, the velocity disturbance is obviously dissipated by the derived robust controller $u^{(20)}$, as illustrated in Fig.~\ref{subfig:velocity_final_controller}. In order to show the improvement of robustness during algorithm learning, the actual attenuation levels of iterative controllers in simulation are shown in Fig.~\ref{subfig:attenuation_comparison}. It can be observed that the disturbance attenuation performance is continuously improved by optimizing the attenuation level in the outer-loop iteration. 

In addition, we also consider a moderate-scale mixed traffic system with $n=15$, $m=5$ to demonstrate the control performance (see Fig.~\ref{fig:large_scale}). In this case, a quadratic polynomial function consisting of $900$ terms is employed to approximate the value function, and after $50$ outer-loop iterations, a centralized controller for the $5$ CAVs is obtained. It can be clearly observed in Fig.~\ref{fig:large_scale} that our method enables the CAVs to cooperatively dampen traffic perturbations and smooth traffic flow. 

% \begin{figure}[!htb]
% \centering
% \subfigure[\label{subfig:velocity_all_hdv}]{
%     \includegraphics[width=0.45\textwidth]{figure/velocity_all_hdv.pdf}
% } 
% \\
% \subfigure[\label{subfig:velocity_distributed}]{
%     \includegraphics[width=0.45\textwidth]{figure/velocity_distributed.pdf}
% } 
% \\
% \subfigure[\label{subfig:velocity_centralized}]{
%     \includegraphics[width=0.45\textwidth]{figure/velocity_centralized.pdf}
% }
% \caption{Simulation results with (a) all HDVs, (b) distributed controller, (c) centralized controller.}
% \label{fig:velocity}
% \end{figure}

\begin{figure}[!tb]
    \centering
    % \captionsetup{singlelinecheck = false,labelsep=period, font=small}
    
    \subfigure[all HDVs]{
        \label{subfig:velocity_all_hdv}
        \begin{minipage}[t]{0.46\linewidth}
        \includegraphics[width=0.99\textwidth]{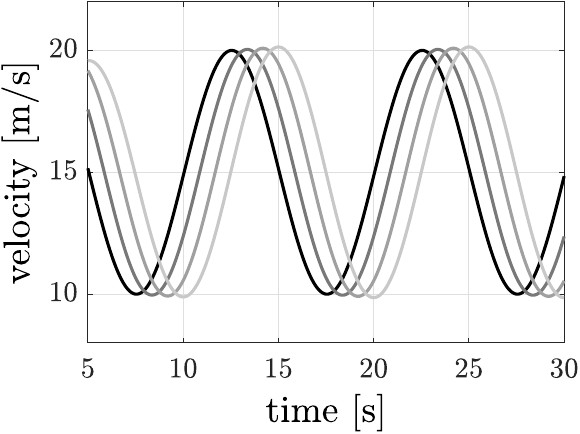}
        \end{minipage}
    }
    \subfigure[Controller after $1$ iteration]{
        \label{subfig:velocity_initial_controller}
        \begin{minipage}[t]{0.46\linewidth}
        \includegraphics[width=0.99\textwidth]{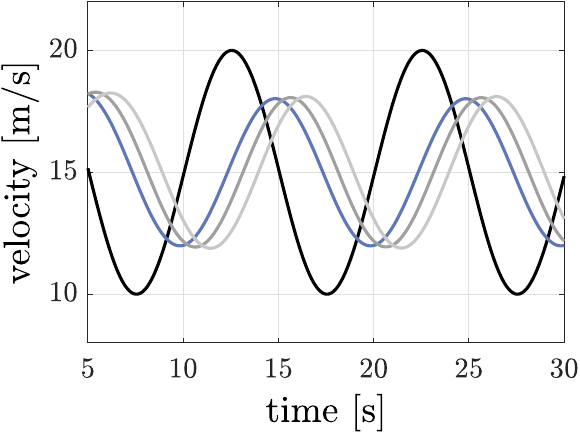}
        \end{minipage}
    }
    \qquad
    \subfigure[Controller after $20$ iterations]{
        \label{subfig:velocity_final_controller}
        \begin{minipage}[t]{0.46\linewidth}
        \includegraphics[width=0.99\textwidth]{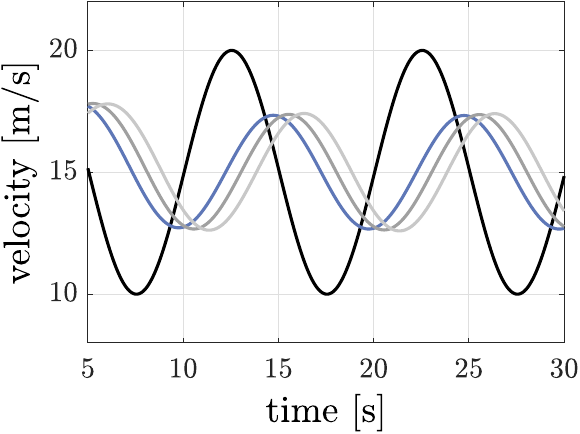}
        \end{minipage}
    }
    \subfigure[Attenuation comparison]{
        \label{subfig:attenuation_comparison}
        \begin{minipage}[t]{0.46\linewidth}
        \includegraphics[width=0.99\textwidth]{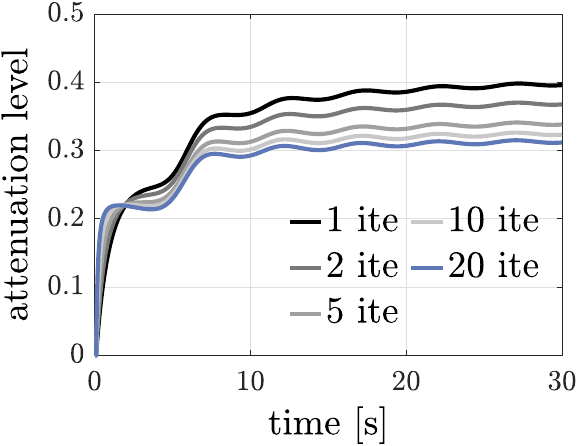}
        \end{minipage}
    }
    \vspace{-2mm}
    \caption{Small-scale simulation results with $n=3$, $m=1$. The black, gray and blue profiles represent the velocity of the head vehicle, the HDVs and the CAV, respectively.  (a) Simulation results when all the vehicles are HDVs. (b)(c) Simulation results under the learned control policies after $1$ or $20$ iterations, respectively. (d) The attenuation level $\gamma$ during the simulations under the controller after different iteration numbers. }
    \label{fig:velocity}
    \vspace{-1mm}
\end{figure}

\begin{figure}[!tb]
    \centering
    \includegraphics[width=0.4\textwidth,trim=0 0 0 30,clip]{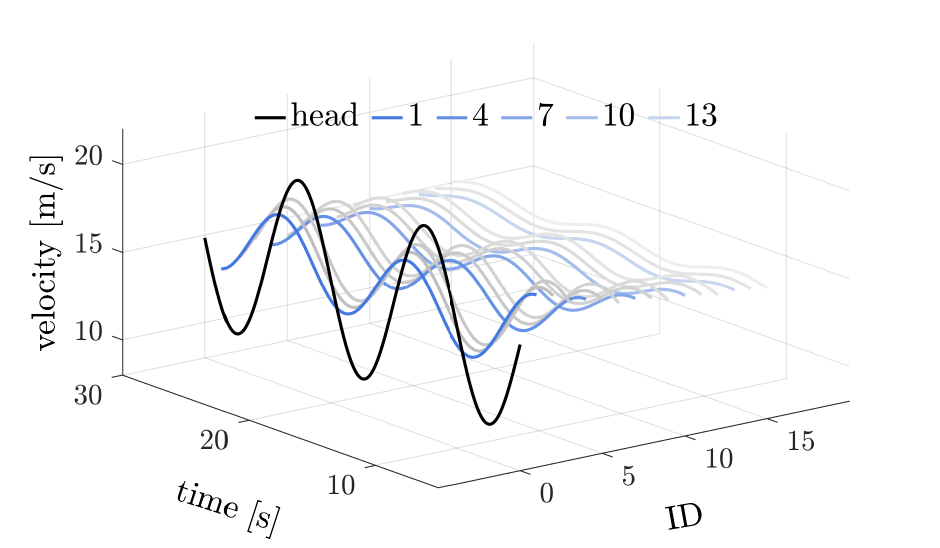}
    \vspace{-2mm}
    \caption{Moderate-scale simulation results with $n=15,m=5$. The meaning of different profiles is consistent with that in Fig.~\ref{fig:velocity}.}
    \label{fig:large_scale}
    \vspace{-4mm}
\end{figure}

%%%%%%%%%%%%%%%%%%%%%%%%%%%%%%%%%%%%%%%%%%%%%%%%%%%%%%%%%%%%%%%%%%%%%%%%%%%%%%%%
\section{Conclusion}
\label{sec.conclusion}

This work investigates the optimal robust control problem of nonlinear mixed traffic systems. In order to reduce the influence of external disturbances from the head vehicle on the entire traffic flow, a zero-sum game is first formulated to optimize the worst-case performance. The converted Hamilton-Jacobi inequality is employed to derive robust controllers and reserve space for the optimization of disturbance attenuation performance. A model-based learning algorithm is then presented, combining inner-loop policy iterations and outer-loop attenuation level optimization.
%algorithm is then developed for numerically solving the minimized attenuation level and corresponding control policy. %Finally, the algorithm is implemented on the mixed traffic system through SOSTOOLS toolbox. 
Simulation studies verify the effectiveness of the obtained control policy for the CAVs to mitigate traffic waves. 
Considering possible traffic model mismatches, one future direction is to design similar policy iteration algorithms to address the corresponding robust performance problem. Another interesting topic is to extend the presented method to a model-free learning version, which does not require any priori knowledge of nonlinear mixed traffic dynamics.% or study robust performance problems regarding model parameter

\bibliographystyle{IEEEtran}
\bibliography{IEEEabrv,ref}

\end{document}